\documentclass[twocolumn,squeezed,noversion]{cdcarticle}

\def\MODE{3}

\usepackage{amsmath}
\usepackage{amsfonts}
\usepackage{amssymb}
\usepackage{setspace}
\usepackage{subfigure}
\usepackage{color}
\usepackage{float}
\usepackage{graphicx}
\usepackage{lettrine}

\usepackage[hidelinks]{hyperref}
\usepackage{enumerate}
\usepackage{tikz}
\usepackage{math}

\newtheorem{theorem}{Theorem}
\newtheorem{remark}{Remark}

\newcommand{\bmtx}{\begin{bmatrix}}
\newcommand{\emtx}{\end{bmatrix}}
\newcommand{\bsmtx}{\left[ \begin{smallmatrix}} 
\newcommand{\esmtx}{\end{smallmatrix} \right]}

\graphicspath{{./Figures/}}

\newcommand{\bmatarray}[1]{\left[\begin{array}{#1}}
\newcommand{\ematarray}{\end{array}\right]}

\begin{document}
	
	\title{Control Interpretations for First-Order Optimization Methods}
	
	\if\MODE3\author{Bin Hu \and Laurent Lessard}
	\else\author{Bin Hu\footnotemark[1] \and Laurent Lessard\footnotemark[1]$^,$\footnotemark[2]}\fi
	
	\if\MODE2
	\note{}
	\else
	\note{\phantom{Tp}\vspace{-1pt}}
	\fi
	\maketitle

	\if\MODE3\else
	\footnotetext[1]{B.~Hu and L.~Lessard are with the Optimization Group at the Wisconsin Institute for Discovery.}
	\footnotetext[2]{L.~Lessard is also with the Department of Electrical and Computer Engineering, University of Wisconsin--Madison, Madison, WI~53706, USA. Email: \texttt{\{bhu38, laurent.lessard\}@wisc.edu}}
	\fi

\begin{abstract}
First-order iterative optimization methods play a fundamental role in large scale optimization and machine learning. This paper presents control interpretations for such optimization methods. First, we give loop-shaping interpretations for several existing optimization methods and show that they are composed of basic control elements such as PID and lag compensators. Next, we apply the small gain theorem to draw a connection between the convergence rate analysis of optimization methods and the input-output gain computations of certain complementary sensitivity functions. These connections suggest that standard classical control synthesis tools may be brought to bear on the design of optimization algorithms.
\end{abstract}

\section{Introduction}\label{Introduction}

First-order iterative optimization methods have been widely applied in data science and machine learning~\cite{bubeck2015,shalev2014}. These methods only require access to first-order derivative information, and iterate on the data until satisfactory convergence is achieved. For example, the gradient method is
\begin{equation}\label{eq:FG}
x^{k+1} = x^k - \alpha \grad f(x^k).
\end{equation}
Such simple methods are often favored over higher order methods such as Newton's method when the dimension of the underlying space is large and computing Hessians is prohibitively expensive.

There has been significant recent interest in finding ways to accelerate the convergence of the gradient method while maintaining low iteration costs.
For example, the \emph{Heavy-ball} method includes an additional momentum term
\begin{equation}\label{eq:heavy}
x^{k+1} = x^k - \alpha \grad f(x^k) + \beta(x^k-x^{k-1}).
\end{equation}
This slight modification can yield a dramatic improvement in worst-case convergence rate if $f$ is quadratic. A similar acceleration scheme, \emph{Nesterov's accelerated method}, can improve the convergence rate for strongly convex $f$ with smooth gradients~\cite{Lessard2014,YEN03a}. These convergence results are derived on a case-by-case basis, and the intuition behind the acceleration is still not fully understood.  

Recent efforts have adopted a dynamical system (or differential equation) perspective in analyzing acceleration for convex objectives~\cite{Su2014NIPS,wibisono2016}, though a more general understanding  of acceleration is still lacking (non-convex objectives, inexact computations, etc).
This paper aims to bring new insights on how to accelerate first-order optimization methods for objective functions which are not convex in general. Our main contributions are as follows:
\begin{enumerate}

\item We pose the iterative optimization paradigm as an output regulation problem, which lends itself to a loop-shaping interpretation. In particular, we show that several popular optimization algorithms may be viewed as controllers which are composed of basic PID or lag compensation elements. We also demonstrate that existing parameter tuning guidelines for these optimization methods are consistent with the loop-shaping design guidelines in control theory. 

\item Using the small gain theorem~\cite{zames1966}, we draw a connection between the convergence rate analysis of optimization methods (under sector-bounded assumptions) and the input-output gain computation for a particular complimentary sensitivity function. It follows that the \emph{design} of optimization algorithms for sector-bounded functions can be interpreted as $\mathcal{H}_\infty$ state feedback synthesis. This explains why acceleration typically requires stronger function assumptions (not necessarily convexity)  beyond just sector-bounded gradients.

\end{enumerate}
A related line of research has emerged in the distributed optimization literature~\cite{jakovetic2015,kia2015,wang2010control,wang2011control}. In \cite{wang2010control,wang2011control}, a continuous-time differential equation was used to describe the dynamics of distributed optimization, leading to a natural iterative algorithm which may be interpreted as a PI controller. In \cite{kia2015}, event-triggered control methods were tailored for distributed optimization over networks. In contrast with the work on distributed optimization, the present work is concerned with control-theoretic properties and interpretations of a big class of first-order optimization 
methods.


A second related line of research is the unified integral quadratic constraint framework in~\cite{Lessard2014}, which provides a numerical tool based on semidefinite programming for use in analyzing optimization algorithms. In contrast with this work, the present work uses a small gain approach with a simple interpretation that interfaces with existing results on complementary sensitivity integrals.


The paper is organized as follows. Section~\ref{sec:PF} explains notation and problem formulation, Section~\ref{sec:loopshape} describes our loop-shaping interpretation for first-order methods, and Section~\ref{sec:main} presents our main results involving the small gain theorem and connections to complementary sensitivity.

\section{Preliminaries}
\label{sec:PF}

\subsection{Spaces and operators}
Let $\ell_{2e}^p$ denote the sequences $x\defeq (x^0,x^1,\dots) \subseteq \R^p$, and let $\ell_2^p \subseteq \ell_{2e}^p$ be the set of square-summable sequences, so if $x \in \ell_2^p$, then
$\sum_{k=0}^\infty \| x^k\|^2<\infty$ where
$\|x^k\|^2\defeq (x^k)^\tp x^k$ denotes the
standard Euclidean norm.
We will omit the superscript $p$ when it is implied by context.
%
%
The gain of a causal operator $K: \ell_{2e}\to \ell_{2e}$ is defined as
\begin{align} \label{eq: Lipschitz}
\|K\| \defeq  \sup_{x\in \ell_2; x \neq 0} \frac{\|K x\|}{\|x\|} 
\end{align}
In addition, $K$ is said to be bounded if it has a finite gain.
Notice this gain is induced by $\ell_2$ signals while the operator $K$ itself is defined on $\ell_{2e}$. This definition makes sense since any bounded operator on $\ell_2$ to itself has a natural causal extension to the operator from $\ell_{2e}$ to $\ell_{2e}$.  Clearly, every bounded operator must map zero inputs to zero outputs.

\if\MODE3
\subsection{Various objective functions in optimization}
\else
\subsection{Various classes of objective functions in optimization}
\fi

Consider the unconstrained optimization problem
\begin{align}
\label{eq:minP}
\min_{x\in \R^p} \,\, f(x)
\end{align}
%
where it is assumed that there exists a unique $x^\star\in \R^p$ satisfying $\grad f(x^\star)=0$.
How to solve \eqref{eq:minP} and find $x^\star$ heavily depends on the assumptions about $f$. A simple assumption is that $f$ is quadratic. Two other common assumptions are $L$-\emph{smoothness} and \emph{strong convexity}. A continuously differentiable function $f:\R^p\to \R$ is $L$-smooth if  the following inequality holds for all $x, y\in \R^p$
\begin{align}
\|\grad f(x)-\grad f(y)\|\le L \|x-y\|.
\end{align}
We define $\mathcal{L}(L)$ to be the set of $L$-smooth functions.
The continuously differentiable function $f$ is $m$-strongly convex if  the following inequality holds for all $x,y \in \R^p$
\begin{align}
\label{eq:strongC}
f(x)\ge f(y)+\grad f(y)^\tp (x-y)+\frac{m}{2} \|x-y\|^2.
\end{align}
Note that we recover ordinary convexity in~\eqref{eq:strongC} if $m=0$. We define $\mathcal{F}(m,L)$ to be the set of functions that are both $L$-smooth and $m$-strongly convex.
The class $\mathcal{F}$ covers a large family of objective functions in machine learning, including $\ell_2$-regularized logistic regression~\cite{teo2007}, smooth support vector machines~\cite{lee2001ssvm}, etc. Clearly, $\mathcal{F}(m,L)\subset\mathcal{L}(L)$.
For all $f \in \mathcal{F}(m,L)$, the following inequality holds for all $x\in \R^p$
\begin{align}
\label{eq:gradient3}
\bmtx x-x^\star \\ \grad f(x)\emtx^\tp \bmtx -2mL I_p & (L+m) I_p \\  (L+m) I_p & -2 I_p\emtx  
\bmtx x-x^\star \\ \grad f(x)\emtx\ge 0
\end{align}
where $I_p$ denotes the $p\times p$ identity matrix \cite[Lemma 6]{Lessard2014}.
On the other hand, a function satisfying the above inequality may not belong to $\mathcal{F}(m,L)$, and may not even be convex. The set of continuously differentiable functions satisfying \eqref{eq:gradient3} is denoted as $\mathcal{S}(m,L)$. This class of functions has sector-bounded gradients, and includes $\mathcal{F}(m,L)$ as a subset.

\subsection{Review of first-order optimization methods}

A classical way to solve~\eqref{eq:minP} is the gradient descent method, which uses the iteration~\eqref{eq:FG} to gradually converge to $x^\star$.
%
%
The intuition behind gradient descent method is as follows. At each step $k$, we find a quadratic approximation of $f$ about $x^k$, which hopefully captures the local structure of $f$, and we solve the quadratic minimization problem
\begin{align}
\min_{x\in \R^p}\left(f(x^k)+\grad f(x^k)^\tp (x-x^k)+\frac{1}{2\alpha}\|x^k-x\|^2\right).
\end{align}
%
When $f\in \mathcal{S}(m,L)$, if $\alpha$ is chosen well, then there exists a constant $\rho\in (0,1)$ and a constant $c\ge 1$ such that
\begin{align}
\|x^k-x^\star\|\le c \rho^{k} \|x^0-x^\star\|
\end{align}
Thus the iterates $\{x^k\}$ converge exponentially to $x^\star$. By convention, this is known as \emph{linear convergence} in the optimization literature.
For example, we can choose $\alpha=\frac{2}{L+m}$ and obtain $\rho=\frac{L-m}{L+m}$ and $c=1$. Another popular choice is $\alpha=\frac{1}{L}$, which leads to $\rho=1-\frac{m}{L}$ and $c=1$. These results are formally documented in \cite[Section 4.4]{Lessard2014}. It is emphasized that the proofs of these results only require \eqref{eq:gradient3}.

When $f\in \mathcal{F}(m, L)$, one can achieve a better convergence rate $\rho=\sqrt{1-\sqrt{\frac{m}{L}}}$ using Nesterov's accelerated method:
\begin{align}\label{eq:NFG}
\begin{aligned}
x^{k+1}&=y^k-\alpha \grad f(y^k)\\
y^k&=(1+\beta) x^k-\beta x^{k-1}
\end{aligned}
\end{align}
where $\alpha=\frac{1}{L}$ and $\beta=\frac{\sqrt{L}-\sqrt{m}}{\sqrt{L}+\sqrt{m}}$. When $L/m$ is large, Nesterov's accelerated method guarantees a much faster convergence rate compared to the gradient descent method. This fact was stated in \cite[Theorem 2.2.3]{YEN03a}.

When $f$ is a quadratic function, one can accelerate the gradient descent method by incorporating a momentum term into the iteration, such as the Heavy-ball method~\eqref{eq:heavy}. Although the Heavy-ball method works extremely well for quadratic objective functions, it can fail to converge for other functions in $\mathcal{F}(m,L)$; see \cite[Section 4.6]{Lessard2014}.

The intuitions behind Nesterov's accelerated method and the Heavy-ball method are still not fully understood. Hence, there is no intuitive way to modify these methods to accelerate the convergence when optimizing more general functions, i.e. $f\in\mathcal{S}(m,L)$ or $f\in\mathcal{S}(m,L)\cap\mathcal{L}(L)$. Gaining intuition for these methods can be beneficial for designing accelerated schemes for more general classes of objective functions.

Finally, it is worth mentioning that every optimization method mentioned in this section can be cast as the feedback interconnection $F_u(P, K)$ as shown in Fig~\ref{fig:fdbd}. Here, $P \defeq  \grad f$ is a static nonlinearity and $K$ is a linear time-invariant (LTI) system (the algorithm).  

\begin{figure}[h]
\centering
\scalebox{0.9}{
\begin{picture}(172,80)(23,25)
 \thicklines
 \put(80,25){\framebox(30,30){$K$}}
 \put(80,75){\framebox(30,30){$\grad f$}}
 \put(42,64){$u$}
 \put(55,40){\line(1,0){25}}  
\put(55,40){\line(0,1){50}}  
 \put(55,90){\vector(1,0){25}}  
 \put(143,64){$v$}
 \put(135,90){\line(-1,0){25}}  
 \put(135,40){\line(0,1){50}}  
 \put(135,40){\vector(-1,0){25}}    
\end{picture}
} 
\caption{Feedback representation for the optimization method $F_u(\grad f, K)$.}
\label{fig:fdbd}
\end{figure}
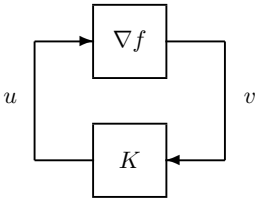

Feedback representations for optimization methods are discussed in \cite[Section 2]{Lessard2014}. In \cite{Lessard2014}, $P$ is a static nonlinear operator that maps $u$ to $v=Pu$ as $v^k=\grad f(u^k)$.  However, this is not a bounded operator since it does not map zero inputs to zero outputs. For the convenience of our discussion, we will choose $P$ to be the operator which maps $u$ to $v=Pu$ as $v^k=\grad f(u^k+x^\star)$ where $x^\star$ is the unique point satisfying $\grad f(x^\star)=0$. Then this choice of $P$ leads to a bounded operator. One can perform a state shifting argument to the feedback representations in \cite{Lessard2014} and cast all the mentioned optimization methods as
\begin{align}
\label{eq:interSS}
\begin{split}
\xi^{k+1}&=A \xi^k +B v^k\\
u^k&=C \xi^k\\
v^k&=\grad f(u^k+x^\star)
\end{split}
\end{align}
where $(A, B, C)$ are the state matrices of $K$.
For example, to rewrite the gradient descent method \eqref{eq:FG}, one can set $\xi^k=u^k=x^k-x^\star$ and $v^k=\grad f(x^k)=\grad f(u^k+x^\star)$. Then \eqref{eq:FG} can be cast as \eqref{eq:interSS} with $(A, B, C)=(I_p, -\alpha I_p, I_p)$.
Notice here $\xi^k=x^k-x^\star$ and hence the convergence rate of the optimization algorithm is equivalent to the rate at which $\xi^k$ goes to $0$.  The optimization method converges to the optimum $x^\star$ at a rate $\rho$ if and only if the model \eqref{eq:interSS} drives $\xi^k$ to $0$ from any initial conditions at the same rate $\rho$.
Using similar arguments, Nesterov's accelerated method and the Heavy-ball method can be written as \eqref{eq:interSS}. In these two cases, the associated state matrices for $K$ are the same as (2.5) and (2.7) in \cite{Lessard2014}, although the states have been shifted by $x^\star$.  When $f\in \mathcal{S}(m, L)$, the inequality \eqref{eq:gradient3} imposes a sector bound on the input/output pair of $P$. Let $v=Pu$. Then the following inequality holds for all $k$
\begin{align}
\label{eq:sectorB}
\bmtx u^k \\ v^k\emtx^\tp \bmtx -2mL I_p & (L+m) I_p \\  (L+m) I_p & -2 I_p\emtx  
\bmtx u^k \\ v^k\emtx\ge 0.
\end{align}
The above inequality is important for further analysis of optimization methods.

\subsection{Input-output stability and small gain theorem}
The key analysis tool in this paper is the small gain theorem, which is now briefly reviewed.
Suppose two causal operators $P : \ell_{2e}\to \ell_{2e}$ and $K : \ell_{2e} \to \ell_{2e}$ both map zero
input to zero output.  Let $[P, K]$ denote the feedback interconnection of $P$ and $K$ illustrated in Fig.~\ref{fig:feedback}:
\begin{align} \label{eq: FB}
\left\{\begin{array}{l}
v = P u + e  \\
u= Kv+ r.
\end{array}\right.
\end{align}

\begin{figure}[h]
\centering
\scalebox{1}{
\begin{picture}(162,66)(10,37)
\thicklines
\put(17,90){$r$} 
\put(10,85){\vector(1,0){30}}
\put(43,85){\circle{5}}
\put(55,90){$u$}
\put(45,85){\vector(1,0){33}}
\put(78,72){\framebox(26,26){$P$}}
\put(104,85){\line(1,0){35}}
\put(139,85){\vector(0,-1){37}}
\put(165,50){$e$} 
\put(172,45){\vector(-1,0){31}}
\put(139,45){\circle{5}}
\put(121,50){$v$}
\put(136,45){\vector(-1,0){32}}
\put(78,32){\framebox(26,26){$K$}}
\put(78,45){\line(-1,0){35}}
\put(43,45){\vector(0,1){38}}
\end{picture}
} 
\caption{Feedback interconnection with exogenous inputs}
\label{fig:feedback}
\end{figure}
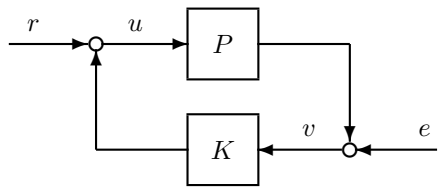

The interconnection $[P, K]$ is said to be \emph{well-posed} if the map $(u, v) \mapsto (r, e)$ defined by \eqref{eq: FB} has a causal inverse on $\ell_{2e}$. It is (input-output) \emph{stable} if it is well-posed and this inverse causal map from $(r, e)$ to $(u, v)$ is bounded.
Clearly, $u, v \in \ell_{2}$ for all $r, e\in \ell_2$ if $[P, \Delta]$ is stable.  Well-posedness holds only if the solutions to~\eqref{eq: FB} have no finite escape time. The small gain theorem states the following~\cite{khalil01,zames1966}.

\begin{theorem}[small gain theorem]
Suppose $P$ and $K$ are bounded causal operators and $[P, K]$ is well-posed. If $\|P\| \|K\|<1$, then $[P, K]$ is input-output stable.
\end{theorem}

The small gain theorem can be used to check the input-output stability of $[P, K]$ when the gains of $P$ and $K$ are both known.
Note that
there are exogenous signals $r$, $e$ in the setup of $[P, K]$ and zero initial conditions on $K$ to ensure that $K$ maps the zero input to a zero output. In contrast, $F_u(P, K)$ allows any initial condition for $K$, so the optimization method $F_u(P, K)$ can be initialized at any initial condition $\xi^0\in \R^n$.  For any optimization method $F_u(P, K)$ described by \eqref{eq:interSS}, one can form an associated interconnection $[P, K]$ by adding the signals $(r, e)$ and fixing the initial condition of $K$ to be zero, since the nonlinear static map $P$ is set up in a way to map zero inputs to zero outputs. 
An important connection between the internal stability of $F_u(P, K)$ and the input-output stability of $[P, K]$ has been stated in \cite[Proposition 5]{Lessard2015}. Consequently, one may apply the small gain theorem for the convergence rate analysis of optimization methods.

\section{Loop-shaping interpretations for\\ optimization methods}
\label{sec:loopshape}

This section presents basic control interpretations for the gradient descent method, Nesterov's accelerated method, and the Heavy-ball method with the hope of shedding light on the general principles underlying the design of first-order methods. The goal of the optimization method is to find $x^\star$ satisfying $\grad f(x^\star)=0$. Hence the optimization method may be viewed as a controller that regulates the plant ``$\grad f$" to zero. When viewing $\grad f$ as the plant one wants to control, the unconstrained optimization problem \eqref{eq:minP} is an output regulation problem, and the LTI part $K$ in the first-order optimization method can be viewed as a controller. The key issue for this output regulation problem is that the equilibrium point $x^\star$ is unknown.

Transfer functions for the controller $K$ are listed in Table~\ref{tab:TransK}. We use the symbol $\otimes$ to denote the Kronecker product. These products appear because the controllers corresponding to our algorithms of interest are repetitions of a single-input-single-output (SISO) system.

\begin{table}[!h]
\begin{center}
\caption{Transfer function $K(z)$ for first-order methods}
\label{tab:TransK}
\begin{tabular}{l|l}\hline\rule{0pt}{2.6ex}%
Optimization Method  & Controller $K$ \\\hline\rule{0pt}{5.2mm}%
Gradient Descent & $I_p\otimes \frac{-\alpha}{z-1}$ \\[2mm]
Heavy-Ball & $I_p\otimes \frac{-\alpha z}{z^2-(1+\beta)z +\beta}$\\[2mm]
Nesterov's Method & $I_p \otimes \frac{-\alpha (1+\beta) z+\alpha \beta}{z^2-(1+\beta)z +\beta}$ \\[2mm]
\hline
\end{tabular}
\end{center}
\end{table}

For the gradient descent method, $K$ is a pure integrator. Hence, the gradient descent regulates the nonlinear plant $P$ via pure integral control. Integral action is necessary since the algorithm must converge to $x^\star$, which amounts to having zero steady-state error when $K$ tracks a step input.

The Heavy-ball method~\eqref{eq:heavy} differs from gradient descent in the inclusion of an additional momentum term. This momentum term may be viewed as a lag compensator.
The Heavy-ball method corresponds to the following controller:
\begin{align}
K=I_p\otimes \left(\frac{-\alpha}{z-1}\right) \left(\frac{z}{z-\beta}\right)
\end{align}
The first term provides integral action to ensure zero steady-state error as with the gradient method, while the second term 
is a discrete-time lag compensator. The lag compensation has the net effect of 1) 
boosting low-frequency response by a factor of roughly $\frac{1}{1-\beta}$, which improves the tracking speed of the controller and hence the convergence of the algorithm and 2) attenuating high-frequency response by a factor of roughly $\frac{1}{1+\beta}$. 
It intuitively makes sense that the Heavy-ball method can accelerate convergence for quadratic objectives, since the plant $P$ becomes a linear operator in this case. However, the lag compensator increases the slope of the loop gain near the crossover frequency, which may have a detrimental effect on the robustness of the closed loop.
This qualitative observation is confirmed by the fact that the Heavy-ball method may fail to converge at all if the objective function is relaxed to include more general strongly convex functions~\cite{Lessard2014}.

Unlike the Heavy-ball method, Nesterov's accelerated method performs well when applied to strongly-convex objective functions. A control interpretation is that Nesterov's accelerated method includes derivative control to decrease the slope of the loop gain near the crossover frequency, which significantly improves the robustness of the algorithm for certain classes of nonlinearities. To see the derivative controller in Nesterov's accelerated method, rewrite \eqref{eq:NFG} as
\begin{multline}
y^{k+1}=y^k+\beta(y^k-y^{k-1})-\alpha \grad f(y^k)\\
-\alpha \beta (\grad f(y^k)-\grad f(y^{k-1}))
\end{multline}
The last term is a difference of the plant output $\grad f$, and can be viewed as a derivative control.

Nesterov's method may also be interpreted as lag compensation together with integral action, as in the Heavy-ball case. The corresponding controller is
\begin{align}
K=I_p\otimes\left(\frac{-\alpha}{z-1}\right) \left(\frac{(1+\beta)z -\beta}{z-\beta}\right),
\end{align}
which has a zero at $z=\frac{\beta}{1+\beta}$, and this helps increase the slope of the Bode plot near the crossover frequency.
The control interpretations for different optimization methods are summarized in Table \ref{tab:GI1}.

\begin{table}[!h]
\begin{center}
\caption{Control interpretations for first-order methods}
\label{tab:GI1}
\begin{tabular}{l|l}\hline\rule{0pt}{2.6ex}%
Optimization Methods & Control Structure \\\hline\rule{0pt}{5.2mm}%
Gradient Descent & Integral Control \\[1mm]
Heavy-Ball & Lag $+$ Integral Control\\[1mm]
Nesterov's Method & Lag $+$ PID Control \\[1mm]
\hline
\end{tabular}
\end{center}
\vspace{-2mm}
\end{table}

We now demonstrate that the design of state-of-the-art optimization methods is actually consistent with general loop-shaping principles from control theory. 
The loop-shaping principle states that the low-frequency loop gain should be sufficiently large to ensure good tracking performance while the high-frequency loop gain should be small enough for the purpose of noise rejection. In addition, the slope of the loop gain near the crossover frequency should be flat (typically around $-20$ dB/decade) to assure a proper phase margin and good robustness. A thorough discussion on loop-shaping can be found in standard references \cite{franklin2010, nise2007, ogata2010}.

Given a function $f\in \mathcal{F}(m,L)$, the standard gradient descent stepsize is $\alpha=\frac{1}{L}$, and the standard parameter choice for Nesterov's accelerated method is $\alpha=\frac{1}{L}$ and $\beta=\frac{\sqrt{L}-\sqrt{m}}{\sqrt{L}+\sqrt{m}}$. Other parameter choices, when $f$ is quadratic for example, are documented in \cite[Proposition 1]{Lessard2014}. Fig.~\ref{fig:bodeplot} shows the Bode plots of the resultant controllers $K$ (only the SISO part) for all these parameter choices under the assumption that $m=0.01$ and $L=1$.

The Bode plots are consistent with the properties of these optimization methods when a loop-shaping intuition is adopted. First, the gradient descent method with the standard stepsize $\alpha=\frac{1}{L}$ is known to be slower than other first-order methods when $f$ is quadratic. This is reflected in the Bode plot, which shows that the gradient method has a relatively low gain particularly in the low-frequency region. Using the optimal tuning of $\alpha=\frac{2}{L+m}$ improves the gain slightly.
%

Second, the optimal quadratic tuning for all three methods leads to controllers whose crossover frequencies are roughly at $0.5$\,Hz. Intuitively, such tuning places excessive weight on tracking performance and is very fragile to noise at the output of the plant $P$. This is consistent with the known robustness properties of these methods as well. For example, the gradient method with $\alpha=\frac{1}{L}$ is known to be very robust to the noise in the gradient computation while the gradient method with $\alpha=\frac{2}{m+L}$ is known to be fragile to such noise \cite[Section 5.2]{Lessard2014}. Comparing the high frequency responses of these two cases immediately leads to the same conclusion. Finally, the slope of Bode plot at the crossover frequency supports the fact that Nesterov's method works for a larger class of functions than the Heavy-ball method.

\begin{figure}[!h]
\centering
\includegraphics[width=\linewidth]{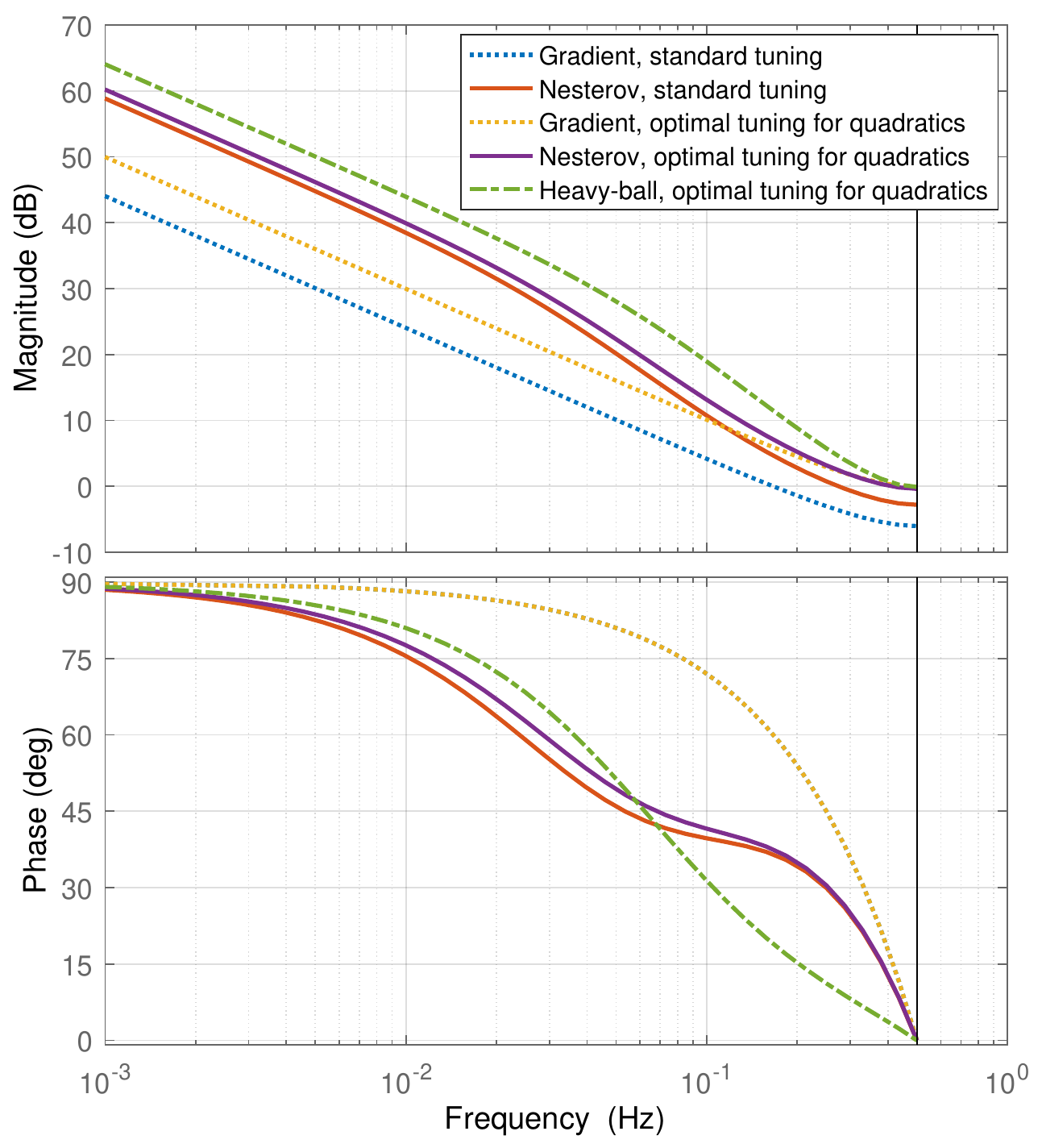}
\caption{Bode plots of $K$ for various first-order methods with commonly-used parameter tunings.}\label{fig:bodeplot}
\end{figure}

In summary, the intuition brought by the traditional loop-shaping theory is consistent with the known properties of the existing first-order methods. This suggests that loop-shaping intuition may be used as a general high-level guideline in the design of optimization methods. The control interpretations above also indicate that classical PID tuning~\cite{ang2005pid} can be used for optimization algorithm design. For example, one could drop the lag compensation and simply use the PID controller:
\begin{align}
K(z)=I_p\otimes \frac{-\alpha(1+\beta)z+\alpha \beta}{z(z-1)}
\end{align}
It remains an open question as to how to choose an appropriate $K(z)$ subject to different assumptions on the objective function. Since acceleration schemes for the optimization of quadratic functions or functions in $\mathcal{F}(m,L)$ already exist,
we will focus on the case $f\in \mathcal{S}(m,L)$.
We will derive one connection between such an optimization design problem and classical control synthesis theory.

\section{Analysis and design of optimization methods using the small gain theorem}
\label{sec:main}

In this section, it is assumed that $f\in \mathcal{S}(m,L)$ and $x^\star$ is the unique point satisfying $\grad f(x^\star)=0$.

\subsection{New feedback representations for first-order methods}
The feedback representation \eqref{eq:interSS} for first-order  methods involves a nonlinear operator $P$ which belongs to the sector $(m,L)$ \eqref{eq:sectorB}. This does not coincide perfectly with a gain bound since upper and lower bounds don't match. We will therefore, use a loop-shifted  
$F_u(P', K')$ that ensures the small gain condition on $P'$ captures the full sector. Choose $P'$ to map $u$ to $v=Pu$ as $v^k=u^k-\frac{2}{m+L}\grad f(u^k+x^\star)$. Then, substitute $\xi^k=u^k=x^k-x^\star$ into the gradient descent method \eqref{eq:FG} to get an alternative feedback interconnection
\begin{align}
\label{eq:newInter}
\begin{split}
\xi^{k+1}&=\left(1-\tfrac{(m+L)\alpha}{2}\right)\xi^k+\tfrac{(m+L)\alpha}{2}v^k\\
u^k&=\xi^k\\
v^k&=u^k-\tfrac{2}{L+m}\grad f(u^k+x^\star)
\end{split}
\end{align}
Direct manipulation of 
\eqref{eq:gradient3} shows that 
$P'$ is in a sector $(-\frac{L-m}{L+m},\frac{L-m}{L+m})$, which leads to gain bound $\|P'\|\le \frac{L-m}{L+m}$. 

Suppose $K=I_p\otimes \bar{K}$ where $\bar{K}$ is a SISO LTI system.
In general, a loop transformation argument can be used to show that any optimization method $F_u(P, K)$ can also be represented as $F_u(P', K')$ where $K'=I_p\otimes \bar{K}'$ and $\bar{K}'=\bar{K}/(\bar{K}-\frac{2}{m+L})$.
Consequently, the feedback interconnection \eqref{eq:newInter} provides another  way to model first-order methods. 

\subsection{Main theorem}

Our main approach is inspired by the loop transformation used in \cite{Lessard2015, BinHu2015a}.
For any $\rho\in(0,1)$, the operators $\rho^+$ and $\rho^-$ are defined as the
time-domain, time-dependent multipliers $\rho^k$, $\rho^{-k}$, respectively. Here, superscripts indicate $k$-th power. 
Define $K_\rho'\defeq\rho^{-} \circ K' \circ \rho^{+}$, and $P_\rho'\defeq\rho^{-} \circ P' \circ \rho^{+}$. From \mbox{\cite[Section 3]{Lessard2015}}, one can conclude $F_u(P', K')$ converges at rate $\rho$ if $[P_\rho', K_\rho']$ is input-output stable and $K_\rho'(z)=K'(\rho z)$. Similarly, define $\bar{K}_\rho'\defeq\rho^{-} \circ \bar{K}' \circ \rho^{+}$, and one has $\bar{K}_\rho'(z)=\bar{K}'(\rho z)$. The main result of this paper is stated below.

\smallskip

\begin{theorem}
\label{thm:main}
Let $\bar{K}'=\bar{K}/(\bar{K}-\frac{2}{m+L})$, and $K'=I_p\otimes \bar{K}'$.
If $\|\bar{K}_\rho'\|<\frac{L+m}{L-m}$, then the optimization method $F_u(P', K')$ has a linear convergence rate $\rho$.
\end{theorem}
\begin{proof}
Notice $P'$ is a pointwise nonlinearity, and hence one can use the small gain condition on $P'$ to show $\|P_\rho'\|\le \frac{L-m}{L+m}$ (see Section 5.1 in \cite{Lessard2015} or Section IV.C in \cite{BinHu2015a} for detailed arguments). In addition, $\|K_\rho'\|=\|I_p\otimes \bar{K}_\rho'\|=\|\bar{K}_\rho'\|<\frac{L+m}{L-m}$, and $\|K_\rho'\|\|P_\rho'\|<1$.
By the small gain theorem, $[P_\rho', K_\rho']$ is input-output stable. By \mbox{\cite[Proposition 5]{Lessard2015}}, $F_u(P', K')$ converges at rate $\rho$.
\end{proof}

\vspace{0.1in}

The power of Theorem~\ref{thm:main} is that it connects the convergence rate analysis of the optimization method to an input-output gain computation on a SISO system $\bar{K}_\rho'$. Note that the $\ell_2$-induced norm~\eqref{eq: Lipschitz} of a stable LTI system is equal to its $\mathcal{H}_\infty$-norm \cite{zhou96}. In addition, we have $\bar{K}_\rho'(z)=\bar{K}'(\rho z)$. Hence we only need to verify that $\bar{K}'(\rho z)$ is stable and then compare the $\mathcal{H}_\infty$-norm of $\bar{K}'(\rho z)$ to $\frac{L+m}{L-m}$.

\subsection{Recovery of rate results for gradient descent}

As a sanity check, we apply Theorem~\ref{thm:main} to recover convergence rate results for the gradient descent method applied to functions in $\mathcal{S}(m,L)$.
Since $\bar{K}(z)=\frac{-\alpha}{z-1}$, we have
\begin{align}
\bar{K}'(z)=\frac{\alpha (m+L)}{2z-2+\alpha (m+L)}
\end{align}
If $\alpha=\frac{2}{m+L}$, it is straightforward to obtain
\begin{align}
\bar{K}'(z)=\frac{1}{z}\,,\quad\bar{K}'(\rho z)=\frac{1}{\rho z}
\end{align}
Clearly, $\bar{K}'(\rho z)$ is stable for any $\rho>0$. Moreover, $\|\bar{K}'(\rho z)\| = \rho^{-1}$.
%
By Theorem~\ref{thm:main}, the gradient descent method converges for any $\rho>\frac{L-m}{L+m}$. This recovers the existing rate result for the gradient method with $\alpha=\frac{2}{L+m}$. 
Another popular choice for $\alpha=\frac{1}{L}$. In this case, the shifted controller is given by
\begin{align}
\bar{K}'(z)=\frac{1+\kappa}{2z-\kappa+1}
\end{align}
where $\kappa\defeq \frac{L}{m}$ is the condition number. Hence, one has
\begin{align}
\label{eq:case2}
\bar{K}'(\rho z)=\frac{1+\kappa}{2\rho\kappa z-\kappa+1}
\end{align}
When $\rho>\frac{1}{2}(1-\frac{1}{\kappa})$, $\bar{K}'(\rho z)$ is stable. In addition, one can substitute $z=1$ into \eqref{eq:case2} to obtain the peak frequency response (the $\mathcal{H}_\infty$ norm) of $\bar{K}'(\rho z)$. To ensure this norm is smaller than $\frac{L+m}{L-m}$, one has the condition
\begin{align}\label{twentythree}
\frac{1+\kappa}{2\rho \kappa -\kappa +1}<\frac{\kappa+1}{\kappa-1}.
\end{align}
Upon simplifying~\eqref{twentythree} together with $\rho>\frac{1}{2}(1-\frac{1}{\kappa})$, we finally obtain $\rho>1-\frac{1}{\kappa}$, which is the linear convergence rate for the gradient descent method when $\alpha=\frac{1}{L}$.

\begin{remark}
A small technical issue in the above analysis is that the rate result proved by the small gain theorem is a strict inequality. This is due to the fact that for LTI systems, input-output stability is slightly stronger than the global uniform stability. See \cite[Remark 1]{BinHu2015a} for a detailed explanation. This issue is negligible from a practical standpoint.
\end{remark}

\subsection{Connections to complementary sensitivity}

Since we have $\bar{K}'(\rho z)=\bar{K}(\rho z)/(\bar{K}(\rho z)-\frac{2}{m+L})$, $\bar{K}'(\rho z)$ is the complementary sensitivity function of the closed-loop system $F_u(\bar{K}(\rho z), -\frac{m+L}{2})$.
Accelerating optimization for $f\in \mathcal{S}(m,L)$ requires finding the smallest $\rho$ such that there exists $\bar{K}(\rho z)$ that stabilizes $\bar{K}'(\rho z)$ while ensuring that $\|\bar{K}'(\rho z)\|  < \frac{L+m}{L-m}$. One possible method would be to perform a bisection search on $\rho$.
For each $\rho$, we can try to design $\bar{K}$ to stabilize $\bar{K}'(\rho z)$ and minimize the $\mathcal{H}_\infty$-norm of $\bar{K}'(\rho z)$ at the same time. This subproblem may be reformulated as an $\mathcal{H}_\infty$ state feedback synthesis problem since $\bar{K}$ always contains a pure integrator
and the state of the scaled dynamics $\frac{1}{\rho z-1}$ is accessible at every timestep. Consequently, the only design variable is the state feedback gain, which happens to be the stepsize of the gradient method. This explains why acceleration in this case is difficult even given memory of past iterates. 
It is worth noting that $\bar{K}(\rho z)$ always has an unstable pole at $z = \rho^{-1}$.
There is a large body of discrete-time complementary sensitivity integral results \cite{middleton1991, mohtadi1990bode, sung1988a, sung1989b} which could potentially be used in studying the design limits of $\bar{K}$ under the analytic constraints posed by the unstable pole at $\rho^{-1}$.

From the above connection, we can see that acceleration typically requires some function properties which can be decoded as constraints involving dynamics. The condition \eqref{eq:gradient3} is a static constraint, and it is hard to design accelerated schemes with this single constraint. However, it is still possible to accelerate non-convex optimization when other function properties are available, e.g. $f\in \mathcal{L}(L)$.

\section{Conclusion}
\label{sec:con}

This paper discussed connections between the analysis of optimization algorithms and classical control-theoretic concepts. Specifically,  the gradient method, the Heavy-ball method, and Nesterov's accelerated method were interpreted as combinations of PID and lag compensators. A loop-shaping interpretation was also used to explain several well-known robustness properties of these algorithms.

We invoked the small gain theorem to show that finding worst-case convergence rates for algorithms amounts to computing the gain of a complementary sensitivity function. In addition, we demonstrated a connection between $\mathcal{H}_\infty$ state feedback synthesis and  stepsize selections of the gradient method. 
These observations are an encouraging first step toward leveraging tools from control theory for the analysis and eventual synthesis of robust optimization algorithms. 


\bibliography{IQCandSOS}


\end{document}